\documentclass[sn-mathphys,Numbered]{sn-jnl}

\usepackage{graphicx}%
\usepackage{multirow}%
\usepackage{amsmath,amssymb,amsfonts}%
\usepackage{amsthm}%
\usepackage{mathrsfs}%
\usepackage[title]{appendix}%
\usepackage{xcolor}%
\usepackage{textcomp}%
\usepackage{manyfoot}%
\usepackage{booktabs}%
\usepackage{algorithm}%
\usepackage{algorithmicx}%
\usepackage{algpseudocode}%
\usepackage{listings}%

\newcommand{\eps}{\epsilon}

\newcommand{\mc}{\mathcal}
\newcommand{\mb}{\mathbb}

\newcommand{\T}{\intercal}

\newtheorem{definition}{Definition}
\newtheorem{assumption}{Assumption}
\newtheorem{proposition}{Proposition}

\newtheorem{theorem}{Theorem}

\newtheorem{problem}{Problem}

\newcommand{\paramZ}{\theta}
\newcommand{\param}{\vartheta}

\begin{document}

\title[Equilibrium Selection in Data Markets: Multiple-Principal, Multiple-Agent Problems with Non-Rivalrous Goods]{Equilibrium Selection in Data Markets: Multiple-Principal, Multiple-Agent Problems with Non-Rivalrous Goods}


\author[1]{\fnm{Samir} \sur{Wadhwa}}

\author[2]{\fnm{Roy} \sur{Dong}}

\affil[1]{\orgdiv{Department of Mechanical Science and Engineering}, \orgname{University of Illinois at Urbana-Champaign}}
\affil[2]{\orgdiv{Department of Electrical and Computer Engineering}, \orgname{University of Illinois at Urbana-Champaign}}


\abstract{There are several aspects of data markets that distinguish them from a typical commodity market: asymmetric information, the non-rivalrous nature of data, and informational externalities. Formally, this gives rise to a new class of games which we call multiple-principal, multiple-agent problem with non-rivalrous goods. Under the assumption that the principal's payoff is quasilinear in the payments given to agents, we show that there is a fundamental degeneracy in the market of non-rivalrous goods. This multiplicity of equilibria also affects common refinements of equilibrium definitions intended to uniquely select an equilibrium: both variational equilibria and normalized equilibria will be non-unique in general. This implies that most existing equilibrium concepts cannot provide predictions on the outcomes of data markets emerging today. The results support the idea that modifications to payment contracts themselves are unlikely to yield a unique equilibrium, and either changes to the models of study or new equilibrium concepts will be required to determine unique equilibria in settings with multiple principals and a non-rivalrous good.}



\maketitle


\section{Introduction}
\label{sec:introduction}


The growing importance of data analytics and machine learning tools for companies to remain competitive has led to data becoming a marketable good in its own right. 
The importance of data analytics has led to the rise of data markets where companies and other entities put forth payments and incentives to obtain data. This can come either in the form of users sharing their personal information, or in the form of workers who exert effort to harvest data. There exist several platforms which act as intermediaries in this exchange, such as Amazon Mechanical Turk, Witkey, and Terbine.

There are several aspects of data markets that distinguish them from a typical commodity market. First, the seller often knows more about the quality of the good than the buyer. In the context of data markets, it is often difficult for a buyer to verify the quality of data: sellers may be hesitant to let buyers `look' at the data beforehand, and, even when in hand, it may be difficult to verify the quality of data. This was an issue early on with Massive Open Online Courses (MOOCs), which scaled by using peer evaluation to assign grades. Course participants often did not exert enough effort to accurately grade their peers and, even with peer evaluations in hand, it was difficult to determine the quality of submitted assignments~\cite{piech2013tuned, miller2005eliciting}. This gives rise to the problem of \textit{asymmetric information}: how can data buyers incentivize data sources to share informative and high-quality data?

Second, data is practically costless to copy, once it has been created or collected. As a result, data sources are able to sell the \textit{same} data to multiple buyers. Exclusivity contracts are often impractical. In many settings, the scale of data collection can make the costs of creating and signing exclusivity contracts intractable. Additionally, in many corporate settings, exclusivity contracts are difficult to enforce because of a competitor's intellectual property walls, which make it difficult to show whether or not a dataset was used. Data is a \textit{non-rivalrous good}, which means that data sources will behave differently in settings with multiple data buyers. It's been shown that non-rivalry can often result in free-riding and social inefficiencies~\cite{morgan-00}. Data markets are no exception~\cite{westenbroek-19}.

Third, when one person shares their data, it can reveal information about others. In the context of data markets, whenever one data source shares their data, it can \textit{change} the informational innovation of every other source's data. Put another way, there are \textit{informational externalities}. Informational externalities can often lead to a decreased market value, since the improvement in inference from other data sources is decreased~\cite{cai-15}. Additionally, it can result in privacy fatalism~\cite{acquisti2015privacy}, since users feel that big corporations `already know everything about them'. This is also sometimes referred to as the `network effects' of data. 

In summary, the three unique characteristics that often arise in data markets are: asymmetric information, non-rivalrous goods, and informational externalities. Previous works in game-theoretic modeling of data markets often consider one or two of these characteristics, but, to the best of our knowledge, this paper is one of the first works that jointly model all three of these facets simultaneously. Models which do not account for asymmetric information often assume the data quality is exogenously fixed, and the data-sharing decision is an opt-in/opt-out decision. Models which do not account for non-rivalry often only have a single data buyer. Models which do not account for informational externalities often only have a single data source, or assume that the data sources are independent.

In settings where asymmetric information, non-rivalry, and information externalities must jointly be considered in the model, this naturally gives rise to a class of problems we refer to as \textbf{multiple-principal, multiple-agent problems with non-rivalrous goods (MPMA-NRG problems)}. Since there is asymmetric information, a data buyer (principal) must incentivize a data source (agent) to exert effort or shed privacy and share accurate data. Since there is non-rivalry, we must consider multiple principals, and since there are informational externalities, we must consider multiple agents. MPMA-NRG problems are a natural class to capture the characteristics of data markets. 

The contribution of this paper is as follows. To the best of our knowledge, this is the first work to identify the class of MPMA-NRG problems and generally study their structure. We look at a general class of payment functions, and show that whenever the principals' payoffs are quasilinear in the payments given, there will be an infinite number of generalized Nash equilibria. Additionally, neither the variational equilibria concept nor the normalized equilibria concept can act as an equilibrium selection method: we show that both these equilibrium concepts will yield a multiplicity of equilibria. These results point at a fundamental ambiguity that exists under very mild conditions for the MPMA-NRG problem.

The rest of this paper is organized as follows. In Section~\ref{sec:background}, we contextualize our contributions among the existing literature. In Section~\ref{sec:model}, we outline our model and formally state our problem of study. Additionally, we show how our model includes some recently studied models for data markets as a special case. In Section~\ref{sec:eq_concepts}, we introduce the equilibrium concepts that will be applied throughout the paper. In Section~\ref{sec:eq_fail}, we prove that multiple-principal, multiple-agent problems with non-rivalrous goods will have an infinite number of generalized Nash equilibria, variational equilibria, and normalized equilibria. This section contains the main theoretical contributions of this paper. We provide closing remarks in Section~\ref{sec:conclusion}.


\section{Background}
\label{sec:background}


The literature relevant to this work can be divided into three broad categories: 1) the statistical and game-theoretic study of data markets, 2) the analysis of markets where asymmetric information arises, and 3) the exploration of equilibrium selection methods.

With regards to the study of data markets, there have been recent works on quantifying the value of data shared. One approach is to consider the usage of data as a coalitional game, where the data points act as the agents and the cooperative goal is some data-driven task~\cite{Kleinberg:2001aa}. This gives rise to metrics for the value of data which are nice theoretically, but often intractable to calculate for large datasets. Recent works have seeked to address this issue using neural networks~\cite{Ghorbani:2019aa} and hashing techniques~\cite{Jia:2019ab, Jia:2019aa}. These works are focused on quantifying the contributions of a data point, and do not consider the data sources as strategic in whether or not they share their information.

Additionally, there have been approaches that consider the strategic behavior of data sources, through game theory and mechanism design. 
In~\cite{Acemoglu:2022aa}, the authors consider a data market with an emphasis on informational externalities and the privacy costs incurred by data sources. This considers settings where the payoffs of the data sources are coupled, but there is only a single data aggregator and the option to share the data is binary (opt-in/opt-out). As a result, the data aggregators do not need to incentivize the data sources to generate informative data; the quality of data is exogenously given. 
In~\cite{cai-15}, the authors consider a setting where multiple data sources interact with a single data aggregator in a statistical estimation task. A peer-prediction based contract mechanism is designed which alleviates issues of asymmetric information. However, there is only a single data aggregator, so the non-rivalrous nature of data is not present. This was extended in~\cite{westenbroek-19} to settings with multiple data aggregators, and was where the multiplicity of equilibria was first shown in the context of data markets.~\cite{westenbroek-19} is the work most relevant to our paper; in this paper, we generalize these results to a broader class of games, and show that common equilibrium selection methods fail in this context.

Relatedly, there are other settings that have motivated recent work in the study of asymmetric information. In these settings, the asymmetric information can be modelled as a principal-agent problem. Broadly speaking, we can divide these works based on whether there is one or more principals. In principal-agent problems with a single principal, the strategic response of agents to a payment contract is studied to design contract mechanisms. These works seek to ensure truthful reporting from agents while maximizing the principal's utility. 
In~\cite{zhang-10}, an electricity market where a single purchaser acquires electricity from multiple sources to satisfy its demand is studied.~\cite{zhang-10} investigates the properties of equilibrium in this scenario and shows how local incentive compatibility (where agents only deviate in a neighbourhood of their true preferences) can be achieved. An electricity market where an electricity producer supplies to a large number of consumers with stochastic demand is studied in~\cite{bitar-13}.
In problems with multiple principals, the added effect of coupling between the decisions of the principals and their effect on the agent's response is studied. In~\cite{kasbekar-14}, spectrum markets with multiple buyers and sellers is modeled as a finite horizon dynamic game and optimal behaviour is determined using a combination of short-term and long-term contracts. 
In~\cite{attar-10}, failure of the revelation principle in the multiple-principal multiple-agent scenario when principals can resort to a complex communication scheme is shown. In~\cite{zeng-19}, the agent's strategy to contract only a subset of principals in order to maximize its utility is studied. 

Lastly, we discuss the work on equilibrium selection methods. 
The literature on equilibrium selection focuses on refinement of equilibrium concepts due to the issue of multiplicity of equilibria that often plagues generalized Nash equilibrium (GNE) problems. Variational equilibrium (VE) and normalized equilibrium (NoE) are the most popular equilibrium selection methods.~\cite{kulkarni-12} provides sufficient conditions for ensuring that VE is a refinement to GNE. In~\cite{facchinei-07} sufficient conditions for equivalence of VE and NoE are presented, however these conditions do not necessarily hold in our model and we treat the two as independent selection methods. Quasi-Nash equilibrium and constrained Nash equilibrium as selection methods are introduced in~\cite{pang-11} and~\cite{kulkarni-09}.~\cite{kulkarni-09} shows the relationship of a constrained Nash equilibrium with VE and GNE.~\cite{huang-13} shows that under the linear independent constraint qualification, the quasi-Nash equilibrium coincides with NoE for their model.

Our work is aimed at studying the consequences of non-rivalrous nature of data in realistic data markets with multiple buyers and multiple sellers. We generalize this study to markets with asymmetric information and any non-rivalrous good and show that the nature of the good leads to an infinite number of generalized Nash equilibria. Furthermore, this set cannot be refined by the concepts of variational equilibrium or normalized equilibrium. 

We model this scenario as a two-stage game. In the first stage, principals decide on contract parameters. In the second stage, agents decide their efforts based on the contracts available to them. We show that for all quasilinear payment contracts that induce a unique dominant strategy equilibrium amongst the agents, this game does not admit a unique solution and the principal bearing the cost of the public good is uncertain. Payment contracts of this form are commonly used in crowdsourcing~\cite{westenbroek-19},~\cite{shah-16},~\cite{suliman-19},~\cite{chen-17}. 

Our results substantiate the idea that modification to payment contracts are unlikely to admit a unique equilibrium and changes in modeling techniques or new equilibrium concepts will be required for analysis of settings with multiple principals and a non-rivalrous good.


\section{Model and Problem Statement}
\label{sec:model}


In this section, we introduce our model for multiple-principal, multiple-agent problems with non-rivalrous goods (MPMA-NRG problems) and state the problem of equilibrium selection. 

\subsection{The MPMA-NRG Problem}

First, let's recall the classical (single) principal-agent (PA) problem~\cite{stiglitz1989principal, grossmanpa}. The PA problem is a two-stage game. In the first stage of the game, the principal announces a payment function $p(\cdot)$ to the agent. In the second stage of the game, the agent chooses an effort level $e$ in response. The asymmetric information implies that the principal does not directly get to observe this effort level $e$, but instead sees a (potentially random) observation $y(e)$. Since the principal does not observe $e$, the payment function $p(\cdot)$ must depend on $y(e)$, and not $e$ directly. The principal wants to minimize their loss, which is a function of the value created by the effort as well as the payments they must give. The agent wishes to maximize their utility, which is a function of the payments received and the effort exerted. 
Since the PA problem is a two-stage game, it can be analyzed as a bi-level optimization.

In the classical PA problem, the principal acts in the first stage and the agent acts in the second stage. To generalize this to MPMA-NRG problems, we assume that the (multiple) principals act in the first stage and the (multiple) agents act in the second stage (see Figure~\ref{fig:pa_vs_mpma}).
In the classical principal-agent problem, the principal decides the optimization problem that the agent solves. In the multiple-principal, multiple-agent problem with non-rivalrous goods, play still occurs in two stages: the (multiple) principals play a game to decide which game is induced between the (multiple) agents.

\begin{figure}[!t]
\centerline{\includegraphics[width=\columnwidth]{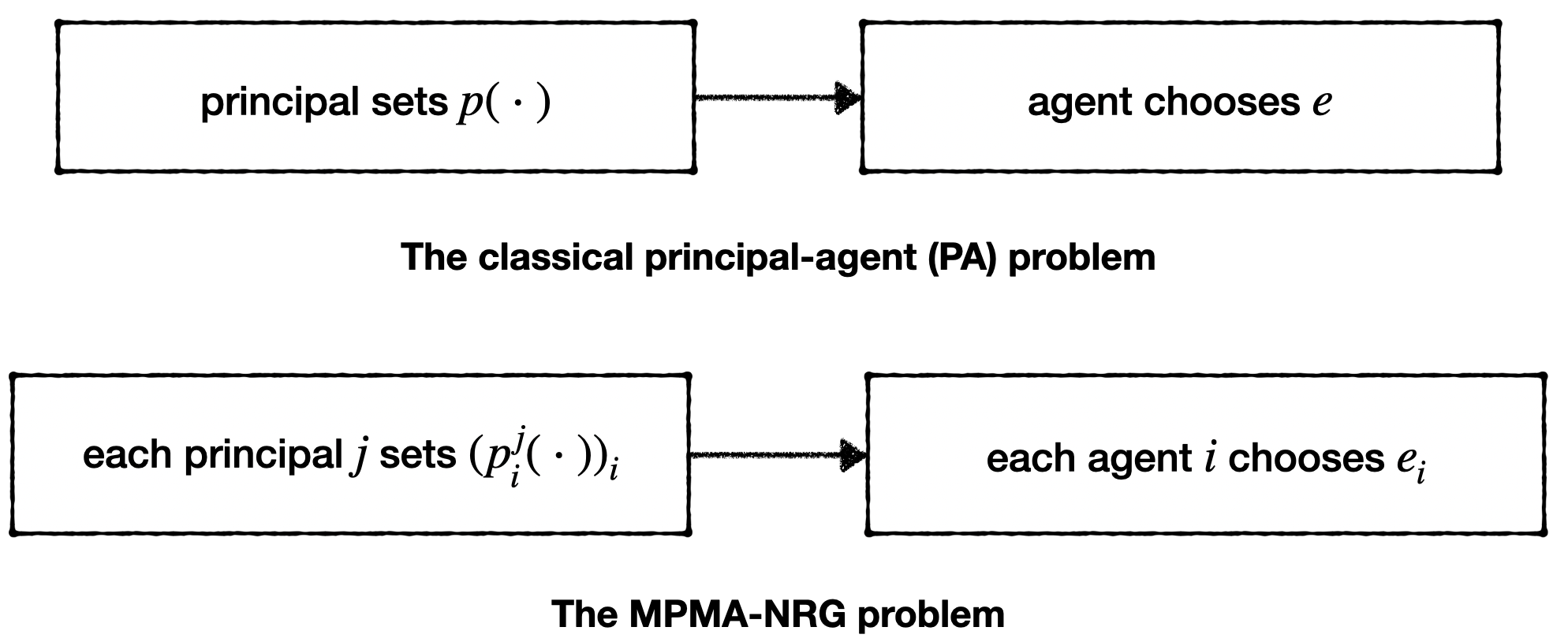}}
\caption{In the classical principal-agent problem, the principal chooses a payment function $p(\cdot)$ and then the agent chooses an effort $e$. In the MPMA-NRG problem, the principals choose their payment functions in the first stage and the agents respond with their efforts in the second stage.}
\label{fig:pa_vs_mpma}
\end{figure}

Formally, let $\mc{M} = \{1,,\dots,M\}$ denote the set of principals and $\mc{N} = \{1,\dots, N\}$ denote the set of agents. We will reserve the index variable $j$ for the principals and $i$ for the agents. 

In the first stage of the game, each principal $j \in \mc{M}$ chooses a payment function for each agent $i \in \mc{N}$, denoted $p_i^j(\cdot)$, resulting in a collection of payment functions $p^j = (p_i^j)_{i \in \mc{N}}$ for each principal $j$. These payment functions are announced simultaneously amongst all principals.

In the second stage of the game, each agent $i \in \mc{N}$ chooses an effort level $e_i \in \mb{R}$ to produce their good, resulting in a vector of efforts $e = (e_i)_{i \in \mc{N}}$, which we will refer to as the \textit{joint effort}. Since the good is non-rivalrous, this effort produces value for each of the principals. We let $v^j(e)$ denote the (potentially random) value created by the joint effort $e$ for principal $j$.

The asymmetric information means that the principal does not directly observe the joint effort $e$. Rather, each principal $j \in \mc{M}$ sees some (potentially random) observable $y^j(e)$. Since this is the information available to the principal, the payment functions should be a function of their available observation $y^j(e)$, rather than the joint effort $e$ directly.

In this paper, we make the following assumptions on the structural form of the principals' losses and the agents' utilities.

\begin{assumption}[Payment structure]
\label{ass:pay_struct}
Each payment function $p_i^j$ is parameterized by $(\paramZ_i^j,\param_i^j)$ as follows:
\[
p_i^j( y^j(e); \paramZ_i^j, \param_i^j ) = \paramZ_i^j + h^j(y^j(e); \param_i^j)
\]
Here, $h^j$ is the incentive structure that the principal $j$ chooses to employ and the $\paramZ_i^j$ parameter is the `base wage' payment which is independent of the observation. For notational simplicity, we will combine $h^j$ and $y^j$ into a single function (potentially random) function $f^j$:
\[
p_i^j( y^j(e); \paramZ_i^j, \param_i^j ) = \paramZ_i^j + f^j(e; \param_i^j)
\]
Note that this is for notational simplicity and should \textbf{not} be understood as the the payments depending directly on the effort. 
We assume $\textbf{f}^j : (e; \param_i^j) \mapsto \mb{E}[f^j(e; \param_i^j)]$ is a continuously differentiable function.

\end{assumption}

Note that the payment to agent $i \in \mc{N}$ from principal $j \in \mc{M}$ depends on the observable $y^j(e)$, which depends on the joint effort $e$. In other words, the payment received by agent $i$ is affected by the efforts exerted by the other agents. Put another way, the payment contracts induce a game between the agents. Payment contracts of this form are studied in~\cite{westenbroek-19},~\cite{shah-16},~\cite{suliman-19} and~\cite{chen-17}. 

Let $\paramZ^j = (\paramZ_i^j)_i$ and $\param^j = (\param_i^j)_i$ denote the payment contract parameters for principal $j \in \mc{M}$. Similarly, let $\paramZ = (\paramZ^j)_j$ and $\param = (\param^j)_j$ denote the contract parameter vectors across all principals, and let $\paramZ_i = (\paramZ_i^j)_j$ and $\param_i = (\param_i^j)_j$ denote all the parameters of relevance to agent $i$.

Now, let us consider the agent's payoffs. 

\begin{assumption}[Agent model]
\label{ass:agent_model}
We assume that the agent's utility function is quasilinear in efforts and payments received; that is, for each agent $i$, their payoff is given by:
\[
u_i(e;\paramZ_i,\param_i) = \sum_{j \in \mc{M}} p_i^j( y^j(e); \paramZ_i^j, \param_i^j )  - e_i
\]
When selecting $e_i$, the agent knows the functions $p_i^j(\cdot;\paramZ_i,\param_i)$ and $y^j$.
\end{assumption}

As mentioned previously, our model captures a non-rivalrous good, so once agent $i$ exerts effort $e_i$, every principal will derive some value from the effort. Consequently, once agent $i$ exerts effort $e_i$, they will receive payments from all the principals.

Next, let us consider the principal's payoffs. 

\begin{assumption}[Principal's loss function]
\label{ass:principal_model}
The loss of each principal $j$ is given by:
\[
L_j(\paramZ^j, \param^j; e) = \sum_{i \in \mc{N}} p_i^j( y^j(e); \paramZ_i^j, \param_i^j ) - v^j(e)
\]
Here, $v^j$ is a (potentially random) function, representing the total value created from joint effort $e$. We assume $\bold{v}^j : e \mapsto \mb{E}[v^j(e)]$ is a continuously differentiable function.
\end{assumption}

The loss for each principal is the sum of all payments made by this principal to the agents, minus the value received when the agents exert effort $e$. In a data markets setting, $v^j(e)$ can represent the quality of statistical inference when the reported dataset is $y^j(e)$.

\begin{assumption}[Risk neutrality]
\label{ass:risk_neut}
We assume that the principals and agents are risk neutral and make their decisions ex-ante. In other words, the principals must decide $(\paramZ^j,\param^j)$ and the agents must decide $e$ prior to the realization of the observables $(y^j(e))_j$ (and thus prior to the realization of the payment values) as well as $(v^j(e))_j$, the value derived from the goods.

As all parties are risk neutral, their ex-ante decisions are made to optimize the expected value of their cost:
\begin{equation}
\label{eq:agent_exp_util}
\mb{E}[u_i(e;\paramZ_i,\param_i)] = \sum_{j \in \mc{M}} \left( \paramZ_i^j + \bold{f}^j(e; \param_i^j) \right) - e_i
\end{equation}
\begin{equation}
\label{eq:princ_cost}
\mb{E}[L_j(\paramZ^j, \param^j; e)] = \sum_{i \in \mc{N}} \left( \paramZ_i^j + \bold{f}^j(e; \param_i^) \right) - \bold{v}^j(e)
\end{equation}
\end{assumption}

Finally, we outline the constraints on the payment contracts the principals may choose from. These constraints ensure the payment contracts induce stable behavior in the agents, and that they willingly participate. We define the desirable properties of a collection of payment contracts here.

First, we consider the stability of chosen effort levels $e^*$. In our case, we enforce that $e^*$ must be a dominant strategy equilibrium between the agents. Intuitively, this means that each agent's payoff is maximized by choosing $e_i^*$, regardless of the effort levels $e_{-i}$ chosen by other agents. 

\begin{definition}[Dominant strategy equilibrium]
For the second stage, we say $e^*$ is a dominant strategy equilibrium for payment contract parameters $(\paramZ, \param)$ if for any $i$ and any other potential effort vector $e'$:
\begin{equation}
\label{eq:dom_strat}
\mb{E}[u_i(e_i^*, e_{-i}';\paramZ_i,\param_i)] \ge
\mb{E}[u_i(e';\paramZ_i,\param_i)]
\end{equation}
We say $e^*$ is unique if no other $e'$ satisfies this property.\footnote{Here, we use the common notational convention where we can single out the $i$th agent's action as $e_i$ in $e = (e_i, e_{-i})$. Although this is a slight abuse of notation, context will make it clear which agent is under consideration.}
\end{definition}

 Next, we wish to have voluntary participation of the agents. Formally, this means that their ex-ante utilities should be positive. We will refer to `ex-ante individual rationality' as simply `individual rationality' throughout this paper.

\begin{definition}[Individual rationality]
\label{def:ic}
Fix a set of payment contract parameters $(\paramZ,\param)$ with a unique dominant strategy equilibrium $e^*$. 
The payment contracts are {\em individually rational} with respect to $(\paramZ,\param)$ if, for every $i \in \mc{N}$:
\[
\mb{E}[u_i(e^*; \paramZ_i, \param_i)] \ge 0
\]
\end{definition}

Whereas individual rationality incentivizes participation (i.e. the creation of the non-rivalrous good), each individual principal must also incentivize the agents to share the non-rivalrous good with them. If the payments are not positive ex-ante, agent $i$ may still produce the good, but choose not to share it with the $j$th principal. Thus, we have another desirable condition.

\begin{definition}[Ex-ante positive payments]
\label{def:positive}
Fix a set of payment contract parameters $(\paramZ,\param)$ with a unique dominant strategy equilibrium $e^*$. 
The payment contracts are {\em positive} with respect to $(\paramZ,\param)$ if, for all $i \in \mc{N}$ and $j \in \mc{M}$:
\[
\mb{E}[p_i^j(y^j(e^*);\paramZ_i^j,\param_i^j)] \ge 0
\]
\end{definition}

This provides our desired behavior in the second stage of the game, formalized in the following assumption. the payment contracts should induce a unique dominant strategy equilibrium $e^*$, and $e^*$ should be individually rational and have positive payments in expectation. This is formalized in the following assumption.

\begin{assumption}[Constraints on the payment function]
\label{ass:payment_const}
The principals choose their payment contract parameters $(\paramZ,\param)$ subject to the constraints that the contracts $p$ induce a unique dominant strategy equilibrium $e^*$. This equilibrium $e^*$ must be individually rational and have positive payments in expectation.
\end{assumption}

We note the changes from the classical principal-agent problem to the MPMA-NRG problem. In the case where there is a single agent, there is often a constraint of incentive compatibility; in settings with multiple agents, this becomes the constraint of a dominant strategy equilibrium. The individual rationality constraint ensures participation, and is the same as in the classical principal-agent problem. Additionally, since there is a non-rivalrous good, we must introduce the new constraint of positive payments to ensure that each agent willingly shares their non-rivalrous good with each principal that provides payments. 

\subsection{Example: effort-averse data sources in data markets}
\label{sec:ex1}

In this subsection, we show how our framework encapsulates the model of a data market presented in~\cite{westenbroek-19}. In this setting, the principals are data aggregators and the agents are effort-averse data sources. A data source produces data by drawing one data sample from a distribution; the higher the level of effort, the lower the variance of the underlying distribution. The data aggregators have access to the reported data, but not the variance of the distributions that generated the data. 

More formally, in this example, the data aggregators are trying to estimate some function $\phi$. 
Each data source $i \in \mc{N}$ generates a data point at $x_i$:
\[
y_i(e_i) = \phi(x_i) + \eps_i
\]
Here, $\eps_i$ is zero mean noise with variance $\sigma_i^2(e_i)$. The variance $\sigma_i^2(e_i)$ decreases as the effort $e_i$ increases. Furthermore, the feature $x_i$ is common knowledge to all parties, and the created data $y(e)$ is the observable for all principals $j \in \mc{M}$.

In this case, the data aggregators issue payments of the form:
\[
p_i^j(y; \paramZ_i^j, \param_i^j) = \paramZ_i^j - \param_i^j \left| y_i - \widehat \phi_{-i}(x_i; y_{-i}) \right|^2
\]
Here, $\widehat \phi_{-i}$ is the leave-one-out estimator~\cite{cai-15}, where $\phi(x_i)$ is estimated using every other data source's data $y_{-i}$. Intuitively, this setting uses the other data sources $-i$ to generate an estimate of what data source $i$ {should} be reporting. The term on the right decomposes into two independent terms: one which increases the payment when the variance decreases, and another term which agent $i$ has no control over.

In this example:
\[
\bold{f}^j : (e, \param_i^j) \mapsto - \param_i^j \mb{E}\left[\left( y_i(e_i) - \widehat \phi_{-i}(x_i; y_{-i}(e)) \right)^2\right] 
\]
Here, the expectation is taken across the $\eps_i$ terms in the observables $y$. 

From the perspective of the principals, their loss is a combination of three terms: a penalty for poor statistical estimation, the cost of payments issued, and an optional penalty for the estimation quality of competing principals. For simplicity, we consider the estimation loss at a single point $x$:
\begin{equation*}
\begin{aligned}
\mb{E}[L_j( (\paramZ^j, \param^j))] = \mb{E} \left[\left( \phi(x) - \widehat \phi^j(x; y(e)) \right)^2\right] \\
- \sum_{k \neq j} \zeta_k^j \mb{E} \left[\left( \phi(x) - \widehat \phi^k(x; y(e)) \right)^2\right] \\
+ \sum_{i \in \mc{N}} \mb{E}[p_i^j(y(e); \paramZ^j, \param^j)]
\end{aligned}
\end{equation*}
Here, $(\zeta_k^j)_k \ge 0$ are parameters that act as weights which capture the competition between different principals. In particular, if the weight $\zeta_k^j$ is positive, then principal $j$ loses their competitive edge when the estimation of principal $k$ is better. 

Thus, in this example:
\begin{equation*}
\begin{aligned}
\bold{v}^j = - \mb{E} \left[\left( \phi(x) - \widehat \phi^j(x;y(e)) \right)^2\right] 
\\
+ \sum_{k \neq j} \zeta_k^j \mb{E} \left[\left( \phi(x) - \widehat \phi^k(x;y(e)) \right)^2\right]
\end{aligned}
\end{equation*}
Note that $\bold{v}^j$ depends on $e$ through the realized data $y(e)$ and the resulting estimators $(\widehat \phi^k)_{k \in \mc{M}}$. 

We note that there are several other approaches to incentive design for strategic data sources which similarly fall into our framework. For example, in~\cite{shah-16}, rather than using a leave-one-out estimator, the authors assume access to queries of known quality, referred to as `gold standard' queries. These approaches can also be encapsulated in the general formulation of the previous section.

\subsection{Problem Statement}

Our problem statement concerns the first stage of the game. We consider settings where the payment contracts induce a unique dominant strategy equilibrium, determining the interactions between agents. Given this enforced structure on the second stage of the game, how can we specify the behavior of the principals in the first stage? In addition to the constraints on the payment contracts outlined, the contract parameters $(\paramZ^j,\param^j)$ must be chosen to minimize principal $j$'s cost in some sense. 

Importantly, having uniqueness of some equilibrium concept which predicts the behavior of both the principals and agents in the MPMA-NRG problem would provide insight into what outcomes to expect in real world applications such as data markets.

\begin{problem}
What equilibrium selection methods can uniquely identify payment contract parameters $(\paramZ,\param)$ and an induced effort level $e^*$? 
\end{problem}

To be a reasonable model for interactions between strategic parties, the parameters $(\paramZ,\param)$ should form some equilibrium concept in the first stage between principals. Additionally, the resulting payment contracts must satisfy these properties in the second stage: the payment contracts induce a unique dominant strategy equilibrium $e^*$ between agents, and be individually rational with respect to $e^*$, with positive ex-ante payments.

Our paper shows that the most common equilibrium concepts and refinements fail to uniquely select an equilibrium in this problem setting. There is a fundamental ambiguity in multiple-principal, multiple-agent problems with non-rivalrous goods, due to the ambiguity in which principals can free-ride off others, and neither variational equilibrium or normalized equilibria can resolve this ambiguity.


\section{Equilibrium Concepts}
\label{sec:eq_concepts}


With our model in place, in this section, we introduce the equilibrium concepts we consider. In Section~\ref{sec:eq_fail}, we will show that none of these concepts will successfully determine a unique set of parameters $(\paramZ,\param)$.

First, let us introduce some notation. Let $\mc{P}_e(e^*)$ denote the set of payment contract parameters such that $e^*$ is a unique dominant strategy equilibrium, and the payment contracts are individually rational and positive with respect to $e^*$. Similarly, let $\mc{P} = \cup_{e^*} \mc{P}_e(e^*)$ denote the set of payment contracts such that there exists an $e^*$ that satisfies these properties.

We introduce the function $\mu$, which maps payment contract parameters to the dominant strategy effort levels. Formally, let $\mu : \mc{P} \rightarrow \mb{R}^\mc{N}$ denote the mapping from payment contract parameters $(\paramZ,\param)$ to the dominant strategy equilibrium $e^*$, which is well-defined for every element of the domain of $\mu$.

The following proposition states that if a payment contract induces a unique dominant strategy equilibrium between agents, this equilibrium depends only on the $\param$ parameters. That is, the parameters $\paramZ$ have no effect on the value of the dominant strategy equilibrium. 

\begin{proposition}
\label{prop:a_dep}

The mapping from contract parameters $(\paramZ,\param)$ to the dominant strategy equilibrium $e^*$ depends only on $\param$. That is, the function $\mu( \cdot , \param )$ is constant on $\mc{P}$.

\end{proposition}

\begin{proof}
Pick any two points $(\paramZ,\param)$ and $(\tilde \paramZ,\param)$ in $\mc{P}$. Let $e^* = \mu(\paramZ,\param)$. We wish to show that $\mu(\tilde \paramZ, \param) = e^*$.

By combining Equations~\eqref{eq:dom_strat} and~\eqref{eq:agent_exp_util} and the fact that $e^*$ is the unique dominant strategy equilibrium for $(\paramZ,\param)$, the following holds for all other effort vectors $e$:
\[
\sum_{j \in \mc{M}} \bold{f}^j((e_i^*, e_{-i}) ; \param_i^j ) - e_i^* \ge
\sum_{j \in \mc{M}} \bold{f}^j(e ; \param_i^j) - e_i
\]
However, this is exactly the same condition for $e^*$ to be a dominant strategy equilibrium for $(\tilde \paramZ, \param)$. By uniqueness, we have $\mu(\tilde \paramZ, \param) = e^*$, as desired. 
\end{proof}

Intuitively, Proposition~\ref{prop:a_dep} holds because the parameters $\paramZ$ act as a constant shift and do not affect the strategic nature of the game at all. Indeed, one can thinking of this as follows: the inclusion of the $\paramZ$ parameters ensures individual rationality while the inclusion of the $\param$ parameters determine the dominant strategy equilibrium.

In light of Proposition~\ref{prop:a_dep}, we will use the notation $\mu( \param )$ without any loss of generality, dropping the $\paramZ$ argument entirely.

Now, we can introduce the concept of generalized Nash equilibrium. The introduction of $\mu$ allows us to note how the dominant strategy effort levels $e^*$ vary with the $\param$ parameters in the following definition.

\begin{definition}[Generalized Nash equilibrium (GNE)~\cite{facchinei-07}]

A set of contract parameters $(\paramZ,\param)$ is a {\em generalized Nash equilibrium} for the first stage of the game if for every $j$, $(\paramZ^j, \param^j)$ solves the minimization problem:
\begin{equation}
\begin{aligned}
\label{eq:gne_opt}
\min_{(\paramZ^j,\param^j)} &~\mb{E}[L(\paramZ^j,\param^j; \mu(\param) )] \\
\mathrm{s.t.} &~(\paramZ,\param) \in \mc{P}
\end{aligned}
\end{equation}
\end{definition}

In other words, holding the actions of other agents constant, $(\paramZ^j, \param^j)$ minimizes the principal $j$'s cost subject to the constraint that $(\paramZ,\param)$ induce a unique dominant strategy equilibrium that is individually rational and ex-ante positive.

It is not uncommon that games admit an infinite set of GNEs~\cite{kulkarni-09}. This has motivated the refinements of equilibrium concepts. We present two here: the variational equilibrium and normalized equilibrium. 

\begin{definition}[Variational equilibria (VE)~\cite{kulkarni-12}]

We say a set of contract parameters $(\paramZ,\param)$ is a {\em variational equilibrium} of the first stage of the game if $(\paramZ,\param) \in \mc{P}$ and, for any $(\tilde \paramZ, \tilde \param) \in \mc{P}$:
\[
\begin{bmatrix}
    \nabla_{(\paramZ^1,\param^1)} \mb{E}[L_1(\paramZ^1,\param^1; \mu(\param))] \\
    \vdots \\
    \nabla_{(\paramZ^M,\param^M)} \mb{E}[L_M(\paramZ^M,\param^M; \mu(\param))] \\
\end{bmatrix}^\T
[(\tilde \paramZ, \tilde \param) - (\paramZ, \param)] \ge 0
\]
\end{definition}

Whereas the first-order optimality conditions of GNE consider variations of a single principal's actions, the conditions for VE check the first-order optimality conditions for all feasible variations, allowing for multiple principals to change their parameters so long as the variation remains feasible.

Furthermore, using the function $\mu$, we can express the individual rationality and positive payments conditions as a finite number of constraints:
\begin{equation}
\begin{aligned}
\label{eq:p_def}
\mb{E}[u_i(\mu(\param);\paramZ_i,\param_i)] \ge 0 & \text{ for all } i \in \mc{N} \\
\mb{E}[p_i^j(y^j(\mu(\param));\paramZ_i^j,\param_i^j)] \ge 0 & \text{ for all } i \in \mc{N} \text{ and } j \in \mc{M}
\end{aligned}
\end{equation}
Thus, the individual rationality and positive payment constraints in Equation~\eqref{eq:gne_opt} can be thought of as a finite number of inequality constraints, and admits a finite number of Lagrange multipliers. This allows us to introduce another equilibrium refinement method.

\begin{definition}[Normalized equilibria (NoE)~\cite{heusinger-12}]
Let $\gamma \in \mb{R}_+^\mc{M}$. Let $\lambda^j$ denote the KKT multipliers from the Equation~\eqref{eq:gne_opt} for principal $j$. 
We say a GNE is a {\em normalized equilibrium} with weights $\gamma$ if the KKT multipliers $\lambda^j$ satisfy the normalization condition:
\[
\gamma_1 \lambda^1 = \gamma_2 \lambda^2 = \dots = \gamma_M \lambda^M
\]
\end{definition}


\section{Failure of Equilibrium Concepts}
\label{sec:eq_fail}


In this section, we outline why the equilibrium concepts in Section~\ref{sec:eq_concepts} fail to uniquely select an equilibrium in the MPMA-NRG problem. 
The inability of variational equilibria to select a unique GNE in Section~\ref{sec:ve_fail} and the non-uniqueness of normalized equilibria in Section~\ref{sec:noe_fail} are the main results of this paper. In doing so, we highlight fundamental degeneracies when multiple principals attempt to incentivize multiple agents to produce a non-rivalrous good.

\subsection{Non-uniqueness of GNE}

First, we note one property which is the underlying reason why the aforementioned equilibrium selection methods fail. Whenever the principal's loss function is quasilinear in the payments given, the individual rationality constraints are always binding at generalized Nash equilibria. This is given in the next proposition.

\begin{proposition}
\label{prop:ir_bind}
Let $(\paramZ,\param)$ be any generalized Nash equilibrium with associated effort level $e^* = \mu(\param)$. Then the individual rationality constraints are binding, i.e. the following holds for every $i \in \mc{N}$:
\[
\mb{E}[u_i(e^*;\paramZ_i,\param_i)] = \sum_{j \in \mc{M}} \left( \paramZ_i^j + \bold{f}^j(e^*; \param_i^j) \right) - e_i^* = 0
\]
\end{proposition}

\begin{proof}
Suppose there exists an $i \in \mc{N}$ such that the individual rationality constraint is not binding. Define $s_i$ as:
\[
s_i = \sum_{j \in \mc{M}} \left( \paramZ_i^j + \bold{f}^j(e^*; \param_i^j) \right) - e_i^* > 0
\]
Then, any principal $j$ can decrease their loss by changing from $\paramZ_i^j$ to $\paramZ_i^j - s_i$. Thus, $(\paramZ^j,\param^j)$ is not an optimizer of Equation~\eqref{eq:gne_opt}, and cannot be a GNE.
\end{proof}

Now, we will explicitly calculate the simplex of generalized Nash equilibria for the game between principals. It is common for GNE to not be unique, and, in fact, it is common for there to be a manifold of GNEs. This lack of uniqueness is one of the commonly stated issues in the study of GNEs~\cite{kulkarni-09}, and has motivated the study of equilibrium refinements such as VE and NoE.

By Proposition~\ref{prop:ir_bind}, we know that the individual rationality constraints always bind. This can be rewritten for any $i \in \mc{N}$:
\[
\sum_{j \in \mc{M}} \paramZ_i^j = e_i^*
- \sum_{j \in \mc{M}} \bold{f}^j(e^*; \param_i^j) =
\mu(\param)_i
- \sum_{j \in \mc{M}} \bold{f}^j(\mu(\param); \param_i^j)
\]
Let's define this quantity on the right side as $g_i(\param)$:
\begin{equation}
\label{eq:def_g}
g_i(\param) = \mu(\param)_i
- \sum_{j \in \mc{M}} \bold{f}^j(\mu(\param);\param_i^j)
\end{equation}
Thus, by re-arranging terms, we can write:
\[
\paramZ_i^j = g_i(\param) - \sum_{k \neq j} \paramZ_i^k
\]
This allows us to express the base wage term $\paramZ_i^j$ as a function of all the other parameters. In particular, we can plug this equality into each principal's ex-ante cost to remove the dependence on $\paramZ_i^j$.
\begin{equation}
\begin{aligned}
\label{eq:no_c_gne}
\mb{E}[L_j(\param^j; \mu(\param))] = 
\sum_{i \in \mc{N}} \left( g_i(\param) - \sum_{k \neq j} \paramZ_i^k + \bold{f}^j(\mu(\param); \param_i^j) \right) \\ - \bold{v}^j(\mu(\param))
\end{aligned}
\end{equation}
In light of this observation, we can think of the interaction between the principals as having two parts. On one hand, in the $\param$ parameters, they decide which game to induce among agents, and what the equilibrium effort levels $e^*$ should be. On the other hand, in the $\paramZ$ parameters, they divide the expected gains from the agents accordingly. 

\begin{proposition}[Simplex of GNE]
\label{prop:gne_simplex}
Let $(\paramZ,\param)$ be any generalized Nash equilibrium. Then, take any $\tilde \paramZ$ such that:
\begin{equation}
\label{eq:c_constraint}
\tilde \paramZ_i^j = g_i(\param) - \sum_{k \neq j} \tilde \paramZ_i^k
\qquad
\tilde \paramZ_i^j \ge - f^j(\mu(\param);\param_i^j)
\end{equation}
Then $(\tilde \paramZ, \param)$ is also a generalized Nash equilibrium. 

Furthermore, there are no other generalized Nash equilibria with the same $\param$ parameters, i.e. the set of $\tilde \paramZ$ which satisfy Equation~\eqref{eq:c_constraint} are exactly the set of $\tilde \paramZ$ such that $(\tilde \paramZ, \param)$ is a GNE.
\end{proposition}

\begin{proof}
The change from $\paramZ$ to $\tilde \paramZ$ does not change the $\param$ parameters that optimize Equation~\eqref{eq:no_c_gne}, and we can see that Equation~\eqref{eq:c_constraint} ensures that the new $(\tilde \paramZ, \param)$ remains feasible for each principal's optimization. Conversely, the GNE condition will imply Equation~\eqref{eq:c_constraint}.
\end{proof}

\subsection{Failure of variational equilibria as a selection method}
\label{sec:ve_fail}

Now that we have noted that there is generally not a unique generalized Nash equilibria, we will show that variational refinements will not allow us to select among the GNE. 

\begin{theorem}
\label{thm:ve}
If one generalized Nash equilibrium is a variational equilibrium, then all generalized Nash equilibria with the same $\param$ parameters are variational equilibria.
\end{theorem}

\begin{proof}
First, let's define $\tilde f^j$ as:
\[
\tilde f^j(\param) = \sum_{j \in \mc{M}} \bold{f}^j(\mu(\param); \param_i^j) - \bold{v}^j(\mu(\param))
\]
Principal $j$'s ex-ante cost in Equation~\eqref{eq:princ_cost} can be written:
\[
\mb{E}[L_j(\paramZ^j, \param^j; \mu(\param))] = \sum_{i \in \mc{N}} \paramZ_i^j + \tilde f^j(\param)
\]
Now, let $(\paramZ,\param)$ be a generalized Nash equilibria and a variational equilibria. 
We can write the VE condition as follows. For any variation $(\tilde \paramZ, \tilde \param) \in \mc{P}$:
\begin{equation*}
\begin{aligned}
\begin{bmatrix}
    \nabla_{(\paramZ^1,\param^1)} \mb{E}[L_1(\paramZ^1,\param^1; \mu(\param))] \\
    \vdots \\
    \nabla_{(\paramZ^M,\param^M)} \mb{E}[L_M(\paramZ^M,\param^M; \mu(\param))] \\
\end{bmatrix}^\T
[(\tilde \paramZ, \tilde \param) - (\paramZ, \param)] = \\
\sum_{j \in \mc{M}} \sum_{i \in \mc{N}} (\tilde \paramZ_i^j - \paramZ_i^j) + 
\sum_{j \in \mc{M}} \sum_{i \in \mc{N}} 
\left.
\frac{\partial \tilde f^j}{\partial \param_i^j}
\right|_{\param} (\tilde \param_i^j - \param_i^j) \ge 0
\end{aligned}
\end{equation*}

Now, let $(\bar \paramZ, \param)$ be any other GNE with the same $\param$ parameters and take any $(\tilde \paramZ, \tilde \param) \in \mc{P}$. By Proposition~\ref{prop:gne_simplex} and Equation~\eqref{eq:c_constraint}, the following holds for each $i$:
\[
\sum_{j \in \mc{M}} \bar \paramZ_i^j = 
\sum_{j \in \mc{M}} \paramZ_i^j
\]
Thus:
\[
\sum_{j \in \mc{M}} \sum_{i \in \mc{N}} (\tilde \paramZ_i^j - \paramZ_i^j) =
\sum_{j \in \mc{M}} \sum_{i \in \mc{N}} (\tilde \paramZ_i^j - \bar \paramZ_i^j)
\]
Considering the VE condition on $(\bar \paramZ, \param)$, we can see that the condition is the same as the condition for $(\paramZ, \param)$ to be a VE. The desired result follows.
\end{proof}

\subsection{Non-uniqueness of normalized equilibria}
\label{sec:noe_fail}

In this section, we show that normalized equilibria also fail to select a unique generalized Nash equilibrium. Recall the set of constraints were given by Equation~\eqref{eq:p_def}. Thus, the individual rationality constraint is for each $i \in \mc{N}$:
\[
\sum_{j \in \mc{M}} \left( \paramZ_i^j + \bold{f}^j(\mu(\param);\param_i^j) \right) - \mu(\param)_i \ge 0
\]
The positivity constraint is for each $i \in \mc{N}$ and $j \in \mc{M}$:
\[
\paramZ_i^j + \bold{f}^j(\mu(\param);\param_i^j) \ge 0
\]
Let:
\begin{equation*}
\begin{aligned}
\tilde g_i(\paramZ,\param) & = 
- \left( \sum_{j \in \mc{M}} \left( \paramZ_i^j + \bold{f}^j(\mu(\param);\param_i^j) \right) - \mu(\param)_i \right) \\
& = - \sum_{j \in \mc{M}} \paramZ_i^j + g_i(\param)
\end{aligned}
\end{equation*}
Here, $g_i(\param)$ is as defined in Equation~\eqref{eq:def_g}. Similarly, let:
\[
\tilde h_{i,j}(\paramZ,\param) =
-\left( \paramZ_i^j + \bold{f}^j(\mu(\param);\param_i^j) \right)
\]
Thus, our inequality constraints are simply $\tilde g_i(\paramZ,\param) \le 0$ for $i \in \mc{N}$ and $\tilde h_{i,j}(\paramZ,\param) \le 0$ for $i \in \mc{N}$ and $j \in \mc{M}$.

\begin{theorem}
Consider any generalized Nash equilibria where all payments are positive in expectation. This GNE is a normalized equilibria for $\gamma_1 = \gamma_2 = ... = \gamma_N$.
\end{theorem}

\begin{proof}
For notational simplicity, let $x^j = (\paramZ^j,\param^j)$.
The first-order stationarity condition for optimality of principal $j$ is:
\begin{equation}
\begin{aligned}
\label{eq:KKT1}
\nabla_{x^j} \mb{E}[L_j (\paramZ^j, \param^j; \mu(\param))] + 
\sum_{i \in \mc{N}} \lambda_i \nabla_{x^j} \tilde g_i(\paramZ,\param) \\
+ \sum_{i \in \mc{N}} \sum_{j \in \mc{M}} \nu_{i,j} \nabla_{x^j} \tilde h_{i,j}(\paramZ,\param) = 0
\end{aligned}
\end{equation}
Let $\tilde f^j$ be as defined in the proof of Theorem~\ref{thm:ve}. Then, the gradients in Equation~\eqref{eq:KKT1} can be calculated:
\[
\nabla_{x^j} \mb{E}[L_j (\paramZ^j, \param^j; \mu(\param))]^\T = 
\begin{bmatrix}
1_{1 \times N} & 
\frac{\partial \tilde f^j}{\partial \param_1^j} & \dots &
\frac{\partial \tilde f^j}{\partial \param_N^j}
\end{bmatrix}
\]
(Here, $1_{m \times n}$ denotes the matrix of dimension $m \times n$ whose entries are all 1.)
Similarly:
\[
\nabla_{x^j} \tilde g_i(\paramZ,\param)^\T =
\begin{bmatrix}
0 & \dots & \underbrace{-1}_{i^{th} \text{pos}} & \dots 0 &
\frac{\partial g_i}{\partial \param_1^j} & \dots &
\frac{\partial g_i}{\partial \param_N^j}
\end{bmatrix}
\]
Now, consider any generalized Nash equilibria in the relative interior of the simplex defined by Equation~\eqref{eq:c_constraint}. By complementary slackness, at these interior points, $\nu_{i,j} = 0$ for all $i$ and $j$. Thus, solving Equation~\eqref{eq:KKT1} for interior points would yield $\lambda_i = 1$ as the only possible values for all $i$.
\end{proof}


\section{Conclusion}
\label{sec:conclusion}

Motivated by applications to data markets, we introduced the multiple-principal, multiple-agents problem with a non-rivalrous good in this paper. We showed that the payment contracts that have been studied in the literature thus far will, in the presence of multiple principals, leads to a multiplicity of equilibria, arising from ambiguity in which principals can free-ride off the others. This is structurally very different from situations where there is either a single principal or a rivalrous good. We have shown that this multiplicity of equilibria exists even for various refinements of equilibrium concepts. The proofs provide intuition for why we have this fundamental degeneracy.

This implies that most existing equilibrium concepts cannot provide predictions on the outcomes of data markets emerging today. Prior to this work, we believed that this multiplicity of equilibria could be addressed by modifications to the payment contracts. However, in this paper, we showed that this degeneracy holds even for a general class of payment contracts, and the proofs of our theorems in this paper outline technical reasons why. This shows that, in order to understand the behavior of strategic parties in data markets, we may need to explore new equilibrium concepts that provide uniqueness in the settings considered in this paper.


\bibliography{./refs}


\end{document}